\documentclass{article}
\usepackage{kr}

\usepackage{times}
\usepackage{soul}
\usepackage{url}
\usepackage[hidelinks]{hyperref}
\usepackage[utf8]{inputenc}
\usepackage[small]{caption}
\usepackage{graphicx}
\usepackage{tikz}
\usepackage{amsmath}
\usepackage{booktabs}
\usepackage{algorithm}
\usepackage{algorithmic}
\newcommand{\supp}{\textit{supp}}
\newcommand{\sat}{\textit{sat}}
\newcommand{\var}{\textit{var}}

\newcommand{\calC}{\mathcal{C}}
\newcommand{\calF}{\mathcal{F}}

\newcommand{\rk}{rk}

\newcommand{\CNF}{\textup{CNF}}
\newcommand{\NNF}{\textup{NNF}}
\newcommand{\DNNF}{\textup{DNNF}}
\newcommand{\dDNNF}{\textup{d-DNNF}}
\newcommand{\sDNNF}{\textup{s-DNNF}}
\newcommand{\sdDNNF}{\textup{sd-DNNF}}
\newcommand{\wDNNF}{\textup{wDNNF}}
\newcommand{\dwDNNF}{\textup{d-wDNNF}}
\newcommand{\swDNNF}{\textup{s-wDNNF}}
\newcommand{\sdwDNNF}{\textup{sd-wDNNF}}

\newcommand{\AC}{\textup{AC}}

\newcommand{\monAC}{\textup{AC}_{m}}
\newcommand{\DmonAC}{\textup{D-AC}_{m}}
\newcommand{\dDmonAC}{\textup{dD-AC}_{m}}
\newcommand{\sDmonAC}{\textup{sD-AC}_{m}}
\newcommand{\sdDmonAC}{\textup{sdD-AC}_{m}}
\newcommand{\wDmonAC}{\textup{wD-AC}_{m}}
\newcommand{\dwDmonAC}{\textup{dwD-AC}_{m}}
\newcommand{\swDmonAC}{\textup{swD-AC}_{m}}
\newcommand{\sdwDmonAC}{\textup{sdwD-AC}_{m}}

\newcommand{\posAC}{\textup{AC}_{p}}
\newcommand{\DposAC}{\textup{D-AC}_{p}}
\newcommand{\dDposAC}{\textup{dD-AC}_{p}}
\newcommand{\sDposAC}{\textup{sD-AC}_{p}}
\newcommand{\sdDposAC}{\textup{sdD-AC}_{p}}
\newcommand{\wDposAC}{\textup{wD-AC}_{p}}
\newcommand{\dwDposAC}{\textup{dwD-AC}_{p}}
\newcommand{\swDposAC}{\textup{swD-AC}_{p}}
\newcommand{\sdwDposAC}{\textup{sdwD-AC}_{p}}

\usepackage{amssymb}
\usepackage{amsthm}
\newtheorem{proposition}{Proposition}
\newtheorem{lemma}{Lemma}
\newtheorem{theorem}{Theorem}

\newtheorem{definition}{Definition}
\newtheorem{question}{Open Question}

\usepackage{todonotes}

\title{A Compilation of Succinctness Results for Arithmetic Circuits}

\author{
Alexis de Colnet\and
Stefan Mengel\\
\affiliations
CNRS, UMR 8188, Centre de Recherche en Informatique de Lens (CRIL), Lens,
F-62300, France\\
Univ. Artois, UMR 8188, Lens, F-62300, France
\emails
\{decolnet, mengel\}@cril.fr
}

\begin{document}

\maketitle

\begin{abstract}

Arithmetic circuits (AC) are circuits over the real numbers with $0/1$-valued input variables whose gates compute the sum or the product of their inputs. Positive AC -- that is, AC representing non-negative functions -- subsume many interesting probabilistic models such as probabilistic sentential decision diagram (PSDD) or sum-product network (SPN) on indicator variables. Efficient algorithms for many operations useful in probabilistic reasoning on these models critically depend on imposing structural restrictions to the underlying AC. Generally, adding structural restrictions yields new tractable operations but increases the size of the AC. In this paper we study the relative succinctness of classes of AC with different combinations of common restrictions.
Building on existing results for Boolean circuits, we derive an unconditional succinctness map for classes of \emph{monotone} AC -- that is, AC whose constant labels are non-negative reals -- respecting relevant combinations of the restrictions we consider. We extend a small part of the map to classes of \emph{positive} AC. Those are known to generally be exponentially more succinct than their monotone counterparts, but we observe here that for so-called deterministic circuits there is no difference between the monotone and the positive setting which allows us to lift some of our results. We end the paper with some insights on the relative succinctness of positive AC by showing exponential lower bounds on the representations of certain functions in positive AC respecting structured decomposability.
\end{abstract}

\section{Introduction}

Arithmetic circuits (AC) are a circuit model for representing polynomials by giving the order in which their inputs have to be combined by sums and multiplications. Thus, AC are not only very natural representations for real-valued polynomials, but also give programs for computing them; this can e.g.~be traced back to~\cite{Valiant80} who called them  ($+$,$\times$)-programs. Today AC play an important role in artificial intelligence because they encompass several classes of circuits with practical applications in probabilistic reasoning, for instance probabilistic sentential decision diagrams (PSDD)~\cite{KisaBCD14} or sum product networks (SPN) with indicator variables~\cite{PoonD11}. AC are also strongly related to concepts such as AND/OR-circuits~\cite{DechterM07} and Cutset Networks~\cite{RahmanKG14}. When used in probabilistic reasoning, AC always represent non-negative functions and are therefore called (somewhat misleadingly perhaps) \emph{positive AC}. Positive AC constitute a subclass of what in the probabilistic graphical models community is called \emph{probabilistic circuits}~\cite{ChoiVB20}. In the literature, positivity is is often syntactically enforced by assuming that all constants in the computation are non-negative, see e.g.~\cite{Darwiche03,PoonD11}, in which case the AC are called \emph{monotone}. Essentially, compared to their monotone counterparts, positive AC encode programs which allow subtraction as an additional operation. This has no impact on the tractability of most operations performed on the AC~\cite{Dennis16} and it is known already since~\cite{Valiant80} that it can decrease the size of AC exponentially. 

While research on arithmetic circuits in complexity theory focuses almost exclusively on trying to show lower bounds on the size of AC representing notoriously challenging  polynomials like the permanent, see e.g.~\cite{JerrumS82,ShpilkaY10,Raz09}, the goals pursued in artificial intelligence are often different: on the one hand, algorithms for generating AC from other models like Bayesian networks~\cite{ChaviraD08,ChoiKD13,KisaBCD14}, or by learning from data~\cite{LowdD08,RooshenasL16}, are a major focus. On the other hand, it is studied how imposing constraints on the structure of AC can render operations like computation of marginals or of maximum a posteriori hypotheses (MAP) or more complex queries tractable on them~\cite{HuangCD06,VergariCLTB21,KhosraviCLVB19}. 
In this latter line of work, the earliest and most well-studied properties are decomposability (also called syntactic multilinearity), smoothness (also called completeness), and determinism. There is an ongoing effort to find new properties: on the one hand, more restrictive properties to allow new operations, for example \emph{structured decomposability}~\cite{KisaBCD14,DangVB20}, on the other hand, more general properties that are sufficient to ensure tractability of important operations. For instance \emph{weak} decomposability (also called consistency) is a relaxation of decomposability which, if combined with smoothness, allows efficient marginals computation~\cite{PeharzTPD15}. 

While the analysis of more restrictive properties is driven by the prospect of AC to support more operations efficiently and therefore be more useful in practice, the quest for more generic properties is motivated by the succinctness of resulting AC: while generally all classes of AC commonly considered can represent all functions, more general classes should intuitively allow smaller representations.

The trade-off between usefulness and succinctness has also been observed for Boolean circuits in negation normal form (NNF) and attracted a lot of attention there~\cite{DarwicheM02,PipatsrisawatD08,BovaCMS16,AmarilliCMS20}. Indeed, all structural restrictions on AC mentioned above are also defined for NNF, and classes of NNF respecting combinations of restrictions have been studied almost exhaustively. In particular, for NNF, succinctness maps have been drawn that intuitively describe the relative succinctness for the classes of NNF one gets by applying different combinations of restrictions. When it comes to AC, research on lower bounds in complexity theory focused on classes with properties such as bounded-depth, tree-like structure, or multilinearity~\cite{GrigorievK98,Raz09,Raz10,ShpilkaY10} that have deep implications in theory but are not particularly desirable in practice -- with the exception of \emph{syntactic} multilinearity which is in fact decomposability. In comparison to Boolean circuits, the succinctness analysis for classes of arithmetic circuits of practical interest is fairly young and far from complete~\cite{MartensM14,ChoiD17}.

In this paper we initiate a systematic succinctness map for AC modeled after that proposed in~\cite{DarwicheM02} for NNF. We focus on classes of AC with $0/1$-variables that respect decomposability or weak decomposability and possibly determinism and$/$or smoothness. Most of our results deal with classes of \emph{monotone} AC and are obtained by lifting results from the existing succinctness map for NNF. To this end, we observe that understanding the succinctness relations between different classes of monotone AC reduces to understanding that between classes of NNF with analogous restrictions. However, several classes of NNF obtained with the reduction, namely those respecting weak decomposability, have only recently been introduced for $\NNF$~\cite{Akshay0CKRS19} and thus their position in the maps has not been studied. To analyze monotone AC, we thus prove the missing succinctness relations for these classes.
From the map for NNF and the lifting technique, we obtain the complete map linking the eight classes of \emph{monotone} AC one gets combining the different restrictions. In a modest contribution to the understanding of \emph{positive} AC, we show that under particular restrictions, all including determinism, the expressive power of classes of positive AC coincide with that of their monotone counterparts. Thus some succinctness relations in the monotone map easily extend to the positive map. However, for positive AC, several relations between classes remain open. 

Finally, in an effort to motivate further research on the succinctness relations left to prove, we describe a technique to show lower bounds on the size of positive AC. We apply it to prove lower bounds for positive AC with \emph{structured} decomposability, which is the case for e.g.~PSDD~\cite{KisaBCD14}. We stress that all separations between classes that we prove are unconditional (so no ``unless P = NP'' or similar assumptions) and exponential.

\section{Preliminaries}

\subsection{AC and NNF}

An arithmetic circuit (short AC) is defined to be a directed acyclic graph with a single source whose sinks are each labeled with a real number, a $0/1$-variable, or by complemented variables $\overline{x}$, and whose internal nodes each have two successors and are labeled by $+$ or $\times$. 
A Boolean circuit in negation normal form (short NNF) is defined completely analogously to an AC, but the internal nodes are labeled with $\lor $ and $\land$ and the only constants that can appear as sink-labels are $0$ and $1$.
The following definitions are the same for AC and NNF, so we do not differentiate the two settings here.

The sinks of a circuit $C$ are called its inputs. We say that variable $x$ appears with negative (resp. positive) polarity in $C$ if $\overline{x}$ (resp. $x$) labels a sink of $C$. If $g$ is an internal node then we denote by $g_l$ and $g_r$ its left and right successors. We define the size $|C|$ of the circuit as the number of nodes in the underlying graph, which, since the operations are binary, is at most twice the number of edges.

Let $X$ be the variables appearing in $C$. An assignment $a$ to $X$ is a mapping from $X$ to $\{0,1\}$. The weight of $a$, denoted by $w(a)$, is the number of variables it maps to 1. A partial assignment is defined as an assignment to a subset $Y\subseteq X$. In the particular case when $Y = \emptyset$, we have the unique empty assignment denoted $a_\emptyset$. Given a partial assignment $a'$, the circuit obtained by \emph{conditioning} $C$ on $a'$, denoted by $C|a'$, is obtained by replacing in $C$ for all $y \in Y$ all inputs labeled $y$ by $a'(y)$ and all inputs labeled $\overline{y}$ by $1-a'(y)$. Given two assignments $a$ and $a'$  to $X$ and $X'$ such that $a$ and $a'$ are consistent on $X \cap X'$, we let $a \cup a'$ denote the assignment to $X\cup X'$ whose restrictions to $X$ and $X'$ are $a$ and $a'$, respectively. For convenience, a literal $\ell_x \in \{x,\overline{x}\}$ will sometimes be seen as an assignment of $x$ to the value satisfying the literal, so we may write $C | \ell_x$ or $a \cup \ell_x$. 

Given an assignment $a$ to $X$, $C$ computes a value $C(a)$ on $a$ in the obvious way by first conditioning $C$ on $a$ and computing in a bottom-up fashion in $C|a$ the results of the internal nodes by computing the result of the operation they are labeled with for the values computed by their successors. $C(a)$ is the value computed by the source node. The function $f : X \rightarrow \mathbb{R}$, resp.~$f : X \rightarrow \{0,1\}$, computed by $C$ is defined as the function defined by $f(a) = C(a)$ for all assignments~$a$.

For an NNF or AC $C$ over variables~$X$ and a node~$g$ in~$C$, let $\var(g)$ denote the subset of $X$ such that $x \in \var(g)$ if and only if $x$ or $\overline{x}$ labels a sink reachable from $g$. Note that if $g$ is a sink labeled with $x$ or $\overline{x}$ then $\var(g) =\{x\}$ and that if $g$ is labeled with a constant then $\var(g) = \emptyset$. By extension $\var(C)$, sometimes called the \emph{scope} of $C$, denotes the set $\var(s)$ where $s$ is the source of $C$. 

The set of assignments to $\var(C)$ for which an AC~$C$ computes a non-zero value is called the \emph{support} of $C$ denoted by $\supp(C)$. For an NNF, these assignments are called \emph{models}, or satisfying assignments, and we use the more common notation $\sat(C)$ for that case instead of $\supp(C)$. For a node $g$ in $C$ we let $C_g$ be the sub-circuit of $C$ consisting of nodes reachable from $g$. We write $\supp(g)$ for $\supp(C_g)$. Note that when $g$ is an input labeled with a literal $\ell_x$, $\supp(g) = \{\ell_x\}$ and that for constant inputs there is $\supp(0) = \emptyset$ and $\supp(\alpha) = \{a_\emptyset\}$ for any constant $\alpha \neq 0$.

\subsection{Subclasses of AC and NNF}

In applications, in particular probabilistic reasoning, the possible outputs of AC are restricted to be non-negative. Thus we define \emph{positive} AC to be the AC that compute non-negative functions, i.e., for all assignments $a$ of its inputs, a positive AC must return a value greater or equal to $0$. We denote the class of all positive AC by $\posAC$. A proper sub-class of $\posAC$ is that of \emph{monotone} AC, denoted $\monAC$, which are the AC whose constant inputs are all non-negative.

The classes studied in this paper correspond to circuits whose nodes enforce one or more of the properties defined below: smoothness (or completeness), determinism, decomposability and weak decomposability. 

\begin{definition}
An internal node $g$ in a circuit $C$ is called \emph{smooth} when $\var(g_l) = \var(g_r)$ holds.

An AC is called \emph{smooth} (or complete) when all its $+$-nodes are smooth. We denote by $\textup{s-}\AC$ the class of smooth AC. 
\end{definition}

\begin{definition}
An internal node $g$ in a circuit $C$ is called \emph{deterministic} when there is no assignment $a$ such that $a_l \in \supp(g_l)$ and $a_r \in \supp(g_r)$, where $a_l$ and $a_r$ are the restrictions of $a$ to $var(g_l)$ and $var(g_r)$, respectively.

An AC is called \emph{deterministic} when all its $+$-nodes are deterministic. We denote by $\textup{d-}\AC$ the class of deterministic AC. 
\end{definition}

\begin{definition}
An internal node $g$ in a circuit $C$ is called \emph{decomposable} when $\var(g_l) \cap \var(g_r) = \emptyset$ holds.

An AC is called \emph{decomposable} when all its $\times$-nodes are decomposable. We denote by $\textup{D-}\AC$ the class of decomposable AC. \end{definition}

We remark that in the complexity theory literature decomposable AC are often called \emph{syntactically multilinear} AC. 

\begin{definition}
An internal node $g$ in a circuit $C$ is called \emph{weakly decomposable} when, for all $x \in \var(g_l) \cap \var(g_r)$, the variable $x$ appears with a unique polarity under $g$, i.e., either $x$ appears under $g$ or $\overline{x}$ appears under $g$, but not both. 

An AC is \emph{weakly decomposable} when all its $\times$-nodes are weakly decomposable. We denote by $\textup{wD-}\AC$ the class of weakly decomposable AC. 
\end{definition}

Weak decomposability is sometimes referred to as \emph{consistency}, but we avoid using this term here since for Boolean circuits it is often used to mean satisfiability. 

The classes $\textup{s-}\NNF$, $\textup{d-}\NNF$, $\textup{D-}\NNF$ and $\textup{wD-}\NNF$ are defined analogously as subclasses of $\NNF$ by replacing $+$-nodes by $\lor$-nodes, $\times$-nodes by $\land$-nodes. However, we will use the more common notations $\DNNF$ and $\wDNNF$ instead of $\textup{D-}\NNF$ and $\textup{wD-}\NNF$.

We will consider intersections of the classes just introduced. The names for the intersection classes combine the prefixes $\textup{s-}$, $\textup{d-}$, $\textup{D-}$ and $\textup{wD-}$ accordingly. For instance the class of deterministic decomposable NNF is denoted d-DNNF, that of smooth weakly decomposable AC is denoted by swD-AC, and so on. Observe that weak decomposability is a generalisation of decomposability, so the intersection of D-AC with wD-AC (resp. DNNF with wDNNF) is just D-AC (resp. DNNF).

Imposing specific combinations of structural restrictions above often makes operations that are intractable on unconstrained AC tractable. For instance, when AC are used in probabilistic reasoning, queries such as the computation of marginals, maximum a posteriori (MAP) or marginal MAP are tractable for different combinations of the four aforementioned constraints~\cite{PeharzTPD15,ShenCD16,KhosraviCLVB19,ChoiVB20,VergariCLTB21}. It turns out that, for all problems studied so far, interesting combinations all include either decomposability or weak decomposability. So, the classes studied in this paper are summarized as followed:
\begin{align*}
\{\emptyset\textup{, s}\}\{\emptyset\textup{, d}\}\{\textup{D, wD}\}\textup{-}\{\monAC, \posAC, \NNF\}.
\end{align*}
which is interpreted as: the subclasses of positive AC, monotone AC and NNF ($\{\monAC, \posAC, \NNF\}$) that implement decomposability or weak decomposability ($\{\textup{D, wD}\}$), and possibly smoothness ($\{\emptyset\textup{, s}\}$) , and possibly determinism ($\{\emptyset\textup{, d}\}$). So eight classes of NNF and sixteen classes of AC.

Let $X$ be a finite set of $\{0,1\}$ valued variables, every function $f : X \rightarrow \mathbb{R}_+$ has a representation in all these classes of $\monAC$ and $\posAC$. To see this, one can just write $f$ as $f(X) = \sum_{a \in \supp(f)} f(a)1_a(X)$, where $1_a(X)$ is the function returning 1 on assignment $a$ and $0$ otherwise. The terms $f(a)1_a(X)$ are easily encoded in positive AC with only decomposable $\times$-nodes, then building a positive AC computing $f$ and implementing smoothness, determinism and decomposability upon those terms AC is straightforward. Analogously, every Boolean functions on $X$ has a representation in all eight classes of NNF studied. However, in both cases there are often more compact circuits than those just described. 

\subsection{Succinctness}

As just discussed, to compare the different classes of circuits considered here, expressivity is not an issue since all classes are fully expressive in the sense that they can represent all functions. However, as we will see, the size of representations in different classes may differ greatly. 
Since the classes allow polynomial time algorithms for different problems, it is meaningful to compare the minimum size of a circuit computing the same function in different classes. This naturally leads to the introduction of \emph{succinctness} as a means to compare the classes. For a class $\calC$, the size of the minimum circuit computing $f$ in $\calC$ is called the $\calC$-size of~$f$. 

\begin{definition}
For two classes of circuits $\calC_1$ and $\calC_2$, we say that $\calC_1$ is \emph{at least as succinct as} $\calC_2$, written $\calC_1 \leq \calC_2$, if there is a polynomial $p$ such that for all $C_2 \in \calC_2$, there exists $C_1 \in \calC_1$ computing the same function with $|C_2| \leq p(|C_1|)$. 
\end{definition}

Equivalently, $\calC_1 \leq \calC_2$ if, for all functions $f$, the $\calC_2$-size is polynomially bounded by the $\calC_1$-size. We write $\calC_1 < \calC_2$ when $\calC_1 \leq \calC_2$ but $\calC_2 \nleq \calC_1$, and $\calC_1 \simeq \calC_2$ when both $\calC_1 \leq \calC_2$ and $\calC_2 \leq \calC_1$ hold; in this case we say that $\calC_1$ and $\calC_2$ are \emph{equally succinct}. Succinctness is a transitive relation.

\subsection{Term subcircuits}

For a (w)D-AC (resp. a (w)DNNF) $C$ on variables $X$, we define \emph{term subcircuits} of $C$ iteratively. Starting from the source, whenever a $\times$-node (resp. $\land$-node) is encountered, its two successors are added to the subcircuit, and whenever a $+$-node (resp. $\lor$-node) is encountered, exactly one arbitrary successor is added to the subcircuit. As indicated by the name, each term subcircuit encodes a function which is a single term $\alpha \times \ell_{x_1} \ell_{x_2} \cdots \ell_{x_k}$ where $\alpha$ is a constant and $\ell_{x_i} \in \{x_i, \overline{x_i}\}$, $(x_i)_{i \in [k]} \subseteq X$. By distributivity, the sum (resp. the disjunction) of all term subcircuits of $C$ is equivalent to~$C$. For DNNF, term subcircuits are more often called \emph{certificates} or \emph{proof trees}, but term subcircuits of wDNNF are generally not shaped like trees. The following easy lemma is shown in the appendix. 

\begin{lemma}
Let $C$ be a (weakly) decomposable $\AC$ (resp.~$\NNF$) then:
\begin{itemize}
\item if $C$ is smooth, all variables appear in all term subcircuits
\item if $C$ is deterministic, then any two distinct term subcircuits $T$ and $T'$ verify $T \times T' = 0$ (resp. $T \land T' \equiv 0$).
\end{itemize}
\end{lemma}

\section{Succinctness Map for Monotone AC}

\subsection{From Monotone AC to NNF}\label{section:phi}
One attractive approach towards understanding the succinctness relations between classes of AC is lifting the corresponding map for classes of NNF to classes of AC. This is because the map for NNF is quite substantial and well understood by now, so building the map for AC upon it would save us the trouble of many proofs. Here we will show that we can apply this approach for classes of \emph{monotone} AC. 
The idea is that separating the classes of Boolean functions corresponding to the support of monotone AC is enough to separate these classes of AC.

Given a monotone AC $C$, we define a Boolean circuit $\phi(C)$ that has the same underlying graph as $C$ and is obtained by just modifying the labels on the nodes of $C$. Sinks labeled by $x$ or $\overline{x}$ or the constant 0 are unchanged, but sinks labeled by constants different from zero are now labeled by the constant 1. For internal nodes, all $\times$-nodes become $\land$-nodes and all $+$-nodes become $\lor$-nodes. Clearly $\var(C) = \var(\phi(C))$ and, since $C$ and $\phi(C)$ have the same graph, we have $|C| = |\phi(C)|$. The following lemmas are easy to derive. Proofs are deferred to the appendix.

\begin{lemma}\label{lemma:mapping_monAC_to_NNF}
When $C$ is a monotone $\AC$, $\phi(C)$ is an $\NNF$ whose models are $\supp(C)$. Moreover if $C$ is (weakly) decomposable, deterministic, or smooth, then $\phi(C)$ is as well.
\end{lemma}

\begin{lemma}\label{lemma:mapping_NNF_to_monAC}
For every $\NNF$ $D$, there exists an $\AC$ $C$ of size $|D|$ whose support are the models of $D$. Moreover if $D$ is (weakly) decomposable, deterministic, or smooth, then so is~$C$.
\end{lemma}

For a class $\calC$ of AC, we define the class of NNF $\phi(\calC):= \{\phi(C)\mid C\in \calC\}$. Lemma~\ref{lemma:mapping_monAC_to_NNF} and Lemma~\ref{lemma:mapping_NNF_to_monAC} directly yield the following:
\begin{proposition}
Let $\gamma$ be any combination of properties from $\{\textup{s,d,D,wD}\}$, then $\phi(\gamma\textup{-}\monAC) = \gamma\textup{-}\NNF$.
\end{proposition}
For instance $\phi(\monAC) = \NNF$, $\phi(\DmonAC) = \DNNF$, $\phi(\dDmonAC) = \dDNNF$, etc. Moreover, since the circuit size is preserved by $\phi$, the following holds:

\begin{proposition}\label{proposition:succinctness_preserved}
Let $\calC_1$ and $\calC_2$ be classes of monotone AC, then $\calC_1 \leq \calC_2$ if and only if $\phi(\calC_1) \leq\phi(\calC_2)$.
\end{proposition}

Since it is known already that $\sDNNF \simeq \DNNF < \dDNNF \simeq \sdDNNF$~\cite{DarwicheM02}, it follows that $\sDmonAC \simeq \DmonAC < \dDmonAC \simeq \sdDmonAC$, and these relations are \emph{unconditional} (so no ``unless P = NP'' or other complexity theoretic assumptions are needed). Weak decomposability has not been studied as widely as  decomposability for $\NNF$, so we here draw the map with the additional classes $\wDNNF$, $\swDNNF$, $\dwDNNF$ and $\sdwDNNF$. Then, using  Proposition~\ref{proposition:succinctness_preserved}, we will obtain the succinctness map for monotone AC shown in Figure~\ref{fig:map_monotone_AC}.

\begin{theorem}\label{theorem:succinctness_map}
The results of Figure~\ref{fig:map_monotone_AC} hold.
\end{theorem}

\begin{figure}[t]
\centering
\begin{tikzpicture}
\def\x{2.6};
\def\y{2};
\def\dx{0.5*\x};
\def\dy{0.5*\y};

\node (wD) at (0,0) {$\wDmonAC$};
\node (D) at (\x,0) {$\DmonAC$};
\node (dwD) at (\dx,\dy) {$\dwDmonAC$};
\node (dD) at (\x+\dx,\dy) {$\dDmonAC$};
\node (swD) at (0,\y) {$\swDmonAC$};
\node (sdwD) at (\dx,\y+\dy) {$\sdwDmonAC$};
\node (sdD) at (\x+\dx,\y+\dy) {$\sdDmonAC$};
\node (sD) at (\x,\y) {$\sDmonAC$};

\draw[-stealth] (dwD) -- (dD);
\draw[-stealth] (dwD) -- (sdwD);
\draw[double distance=1pt] (dD) -- (sdD);
\draw[double distance=1pt] (sdwD) -- (sdD);

\draw[-stealth] (swD) -- (sdwD);
\draw[-stealth] (sD) -- (sdD);
\draw[white,line width=5pt] (swD) -- (sD);
\draw[double distance=1pt] (swD) -- (sD);

\draw[-stealth] (wD) -- (swD);
\draw[-stealth] (wD) -- (dwD);
\draw[white,line width=5pt] (D) -- (sD);
\draw[double distance=1pt] (D) -- (sD);
\draw[-stealth] (D) -- (dD);
\draw[-stealth] (wD) -- (D);

\end{tikzpicture}
\caption{Succinctness map for monotone AC. An arrow $C_1 \rightarrow C_2$ means that $\calC_1 < \calC_2$. A double line $C_1$ \rotatebox[origin=c]{90}{$\parallel$} $C_2$ means that $\calC_1 \simeq \calC_2$. The absence of connector between two classes $\calC_1$ and $\calC_2$ means either that the succinctness relation is derived from transitivity or that the two classes are incomparable, i.e., $\calC_1 \nleq \calC_2$ and $\calC_2 \nleq \calC_1$.}\label{fig:map_monotone_AC}
\end{figure}
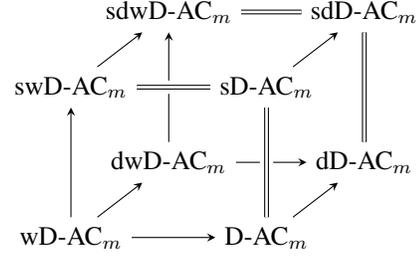

Section 3.3 is dedicated to the proof of Theorem~\ref{theorem:succinctness_map}. But first we will show a useful auxiliary result.

\subsection{Smoothing Restricted $\wDNNF$}

It was shown by Peharz et al.~\shortcite{PeharzTPD15} that transforming general $\wDmonAC$ into $\swDmonAC$ leads to an unavoidable exponential blow-up. By Proposition~\ref{proposition:succinctness_preserved}, the same is true for $\wDNNF$ and $\swDNNF$. Here we show that this is not the case when all term subcircuits have the same variables. 
\begin{proposition}\label{proposition:proof_trees_on_same_variables}
Let $D$ be a $\wDNNF$ over $n$ variables such that for any two term subcircuits $T$ and $T'$, $\var(T) = \var(T')$ holds. Then there is a smooth $\wDNNF$ $D^*$ equivalent to $D$ of size $|D^*|= O(n|D|)$. Furthermore, if $D$ is deterministic, then so is $D^*$.
\end{proposition}
Proposition~\ref{proposition:proof_trees_on_same_variables} will be used in the next section to prove the succinctness map for $\NNF$. Before we prove it, we give some more definitions.

\begin{definition} Let $\ell_x \in \{x,\overline{x}\}$. An $(\land \ell_x)$-link is a $\land$-node whose successors include a leaf labeled by $\ell_x$. Let $g$ be a node and let $p$ be a predecessor of $g$, then inserting an $(\land \ell_x)$-link between $g$ and $p$ means replacing the connection 
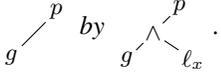
\raisebox{-0.5\height}{
\begin{tikzpicture}
\node[inner sep=1.5,font=\footnotesize] (g) at (0,0) {$g$};
\node[inner sep=1.5,font=\footnotesize] (p) at (0.6,0.6) {$p$};
\draw (g) -- (p);
\end{tikzpicture}
} by 
\raisebox{-0.5\height}{
\begin{tikzpicture}
\node[inner sep=1,font=\footnotesize] (g) at (0,0) {$g$};
\node[inner sep=1,font=\footnotesize] (l) at (0.35,0.35) {$\land$};
\node[inner sep=1,font=\footnotesize] (p) at (0.7,0.7) {$p$};
\node[inner sep=1,font=\footnotesize] (x) at (0.88,0) {$\ell_x$};
\draw (g) -- (l);
\draw (l) -- (p);
\draw (x) -- (l);
\end{tikzpicture}
}.
\end{definition}

We call the intermediate $\land$-node in the construction above the \emph{link node}. A succession of link nodes is a \emph{chain of links}. We remark that links have already been used by Peharz et al.~\shortcite{PeharzTPD15} to analyze the impact of smoothness on $\wDmonAC$, but we use them here in a different way.

For a term subcircuit $T$ containing a node $g$, let $T_g$ denote the sub-circuit of $T$ under $g$, and let $T_{\overline{g}}$ be the sub-circuit of $T$ corresponding to all nodes accessible from the source without passing through $g$. Observe that because of \emph{weak} decomposability, some nodes accessible from $g$ may be reached by paths in $T$ not passing through $g$, so $T_g$ and $T_{\overline{g}}$ are not necessarily disjoint.

\begin{proof}[Proof (of Proposition~\ref{proposition:proof_trees_on_same_variables})]
Let $g$ be an $\lor$-node such that $\var(g_l) \neq \var(g_r)$. Let $x \in \var(g_l)$ and $x \not\in \var(g_r)$. There exists an $\land$-node that is an ancestor of $g$ in $D$, otherwise not all term subcircuits of $D$ would have the same variables.  So, for all term subcircuits $T$ of $D$ containing $g$, $T_{\overline{g}}$ is not empty. Moreover $x$ must be contained in $var(T_{\overline{g}})$ for otherwise we can construct a term subcircuit that does not contain $x$ by extending $T_{\overline{g}}$ to a term subcircuit choosing $g_r$ as the child of~$g$.

We claim that $x$ appears with unique polarity under $g_l$. To see this, assume first that $x$ appears positively in $T_{\overline{g}}$. Now if $\overline x$ appeared below $g_l$ as the label of a node $g^*$. Then, we could extend $T_{\overline{g}}$ to a term subcircuit $T^*$ containing $g^*$ and thus the variable $\overline x$.  But then $T^*$ would contain both $x$ and $\overline x$  which is impossible because $C$ is weakly decomposable. If $x$ appears negatively in $T_{\overline{g}}$, we reason analogously. So in any case, $x$ appears with unique polarity under $g_l$; we assume in the remainder that it appears positively, the other case is completely analogous.

Analogously to above, one sees that for all term subcircuits $T'$ containing $g$, we have that $T'_{\overline{g}}$ contains $x$.
So, for all $T'$ passing through $g$, we have $T \equiv T \land x$. Now insert an $(\land x)$-link between $g$ and $g_r$ and let $D'$ be the resulting $\wDNNF$. We write $(g,g_r) \in T$ when the wire from $g$ to $g_r$ is in the term subcircuit $T$ of $D$. There is a bijection $\varphi$ between the term subcircuits of $D$ and those of $D'$: for a term subcircuit $T$ of $D$, set $\varphi(T) = T$ if $(g,g_r) \not\in T$, and let $\varphi(T)$ be the term subcircuit of $D'$ we get from $T$ by inserting the $(\land x)$-link between $g$ and $g_r$ otherwise. Clearly, when $(g,g_r) \in T$, then $\varphi(T) \equiv T \land x$, and we have already seen that $T \land x \equiv T$ in that case. So  
\begin{align*}
D' &\equiv \bigvee_{T : (g,g_r) \in T} \varphi(T) \,\, \lor \bigvee_{T : (g,g_r)\not\in T} \varphi(T) \\
  &\equiv \bigvee_{T : (g,g_r) \in T} T \indent \,\,\,\, \lor \,\bigvee_{T : (g,g_r)\not\in T} T \indent \indent \equiv D
\end{align*}
Observe that $\var(g)$ is identical in $D$ and $D'$ since $x$ was already in $\var(g_l)$ in $D$. Observe also that the $\land$-link node is decomposable. So $D'$ is a $\wDNNF$. Now in $D'$ the variable $x$ appears under both successors of $g$. We repeat that process until the successors of $g$ have the same set of variables, so until $g$ is smooth. Doing this for all non-smooth $\lor$-nodes yields a $\wDNNF$ $D^*$ that is smooth. The construction only adds chains of links between nodes that were originally in $D$, and since there are $n$ variables, at most $n$ links are inserted between any two connected nodes of $D$, hence $|D^*| = O(n|D|)$.

Finally we argue that if $D$ is deterministic, then so is $D^*$. We just need to prove this for $D'$, i.e., one single addition of~an $(\land x)$-node. Assume that $g$ is deterministic in $D$. Let $g'_r$ be the $\land$-node inserted between $g$ and $g_r$ in $D'$. The successors of $g'_r$ are $x$ and $g_r$. Assume there is an assignment $a'$ to $\var(g)$ whose restrictions $a'_l$ and $a'_r$ to $\var(g_l)$ and $\var(g'_r)$ are in $\sat(g_l)$ and $\sat(g'_r)$ respectively. Then $a'_r$ satisfies $g_r$ so~$g$ is not deterministic in $D$. This is a contradiction, so $g$ remains deterministic in~$D'$ and $D'$ is deterministic.
\end{proof}

\subsection{Proof of Theorem~\ref{theorem:succinctness_map}}

\begin{figure}[t]
\centering
\begin{tikzpicture}
\def\x{2.6};
\def\y{2};
\def\dx{0.5*\x};
\def\dy{0.5*\y};

\node (wD) at (0,0) {$\wDNNF$};
\node (D) at (\x,0) {$\DNNF$};
\node (dwD) at (\dx,\dy) {$\dwDNNF$};
\node (dD) at (\x+\dx,\dy) {$\dDNNF$};
\node (swD) at (0,\y) {$\swDNNF$};
\node (sdwD) at (\dx,\y+\dy) {$\sdwDNNF$};
\node (sdD) at (\x+\dx,\y+\dy) {$\sdDNNF$};
\node (sD) at (\x,\y) {$\sDNNF$};

\draw[-stealth] (dwD) -- (dD);
\draw[-stealth] (dwD) -- (sdwD);
\draw[double distance=1pt] (dD) -- (sdD);
\draw[double distance=1pt] (sdwD) -- (sdD);

\draw[-stealth] (swD) -- (sdwD);
\draw[-stealth] (sD) -- (sdD);
\draw[white,line width=5pt] (swD) -- (sD);
\draw[double distance=1pt] (swD) -- (sD);

\draw[-stealth] (wD) -- (swD);
\draw[-stealth] (wD) -- (dwD);
\draw[white,line width=5pt] (D) -- (sD);
\draw[double distance=1pt] (D) -- (sD);
\draw[-stealth] (D) -- (dD);
\draw[-stealth] (wD) -- (D);

\end{tikzpicture}
\caption{Succinctness map for different subclasses of $\NNF$.}\label{fig:map_DNNF}
\end{figure}
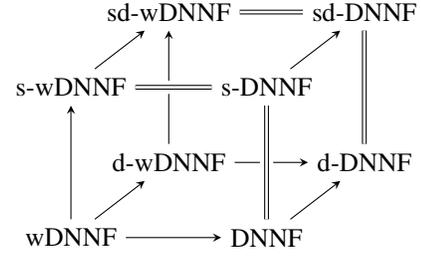

By Proposition~\ref{proposition:succinctness_preserved}, Theorem~\ref{theorem:succinctness_map} is equivalent to proving the correctness of the corresponding map for subclasses of $\wDNNF$ that for the convenience of the reader is given in Figure~\ref{fig:map_DNNF}. So we will exclusively work on that map here and Theorem~\ref{theorem:succinctness_map} follows directly.

It was shown by Darwiche and Marquis~\shortcite{DarwicheM02} that $\sDNNF$ and $\DNNF$ are equally succinct, and that $\sdDNNF$ and $\dDNNF$ are equally succinct, the paper also contains the statement $\DNNF < \dDNNF$ conditioned on standard complexity theoretic assumptions. The result was made unconditional in~\cite{BovaCMS16}. So we already have the right face of the cube-like succinctness map of Figure~\ref{fig:map_DNNF}.

\begin{lemma}
$\wDNNF < \DNNF$.
\begin{proof}
Since $\DNNF \subseteq \wDNNF$ there only is $\DNNF \nleq \wDNNF$ to prove. It is readily verified that monotone NNF, that is, NNF with non-negative literal inputs, are wDNNF. In~\cite{BovaCMS14}, see also \cite{Capelli16}, the separation $\DNNF \nleq \CNF$ is shown finding an infinite class of monotone 2-CNF that have polynomial size but whose equivalent DNNF all have exponential size. Monotone CNF are wDNNF so this proves $\DNNF \nleq \wDNNF$. 
\end{proof}
\end{lemma}

Peharz et al.~\shortcite{PeharzTPD15} give an algorithm to transform any smooth \emph{weakly} decomposable monotone AC into an equivalent smooth decomposable monotone AC in polynomial time\footnote{Peharz et al. work on sum product networks (SPN) with indicator variables inputs. Their SPN differ from our monotone AC in that the non-negative constants are not inputs of the circuit but weights on the connectors of the $+$-nodes. Such SPN are converted into our monotone AC in polynomial time by replacing each weighted edge by a $\times$-node whose successors include the weight.}. Careful examination of the algorithm shows that it can be adapted to turn any $\swDNNF$ into an equivalent $\sDNNF$ in polynomial time (the existence of the transformation actually derives from Lemmas~\ref{lemma:mapping_monAC_to_NNF} and~\ref{lemma:mapping_NNF_to_monAC}). Examining the algorithm even further, one sees that it preserves determinism, so the adapted variant for NNF also gives a polynomial time transformation from $\sdwDNNF$ to $\sdDNNF$. 
 
\begin{lemma}
$\swDNNF \simeq \sDNNF$ and $\sdwDNNF \simeq \sdDNNF$. But $\wDNNF < \swDNNF$.
\end{lemma}
\begin{proof}[Proof sketch]
The proof that $\sDNNF \leq \swDNNF$ and $\sdDNNF \leq \sdwDNNF$ is an adaptation of the techniques in~\cite{PeharzTPD15} to the case of Boolean circuits. The reverse succinctness relations holds since decomposability is a particular kind of weak decomposability.

As for $\wDNNF < \swDNNF$, $\wDNNF \leq \swDNNF$ comes from $\swDNNF$ being a subclass of $\wDNNF$, and $\swDNNF \nleq \wDNNF$ holds for otherwise $\DNNF \nleq \wDNNF$ would be violated by transitivity.
\end{proof}

\begin{lemma}\label{lemma:dednnfvsdnnf}
$\dwDNNF \nleq \DNNF$.
\end{lemma}
\begin{proof}
Consider the class $\calF$ of functions introduced by Sauerhoff \shortcite{Sauerhoff03} and used in \cite{BovaCMS16}. 
All $f \in \calF$ on~$n$ variables have DNNF-size polynomial in $n$ but d-DNNF-size at least $2^{\Omega(\sqrt{n})}$. For an integer~$k$, let $D_k(f)$ be the smallest $\dDNNF$ representing $f \land [w(\cdot) = k]$, i.e., the function whose models are exactly the models of $f$ of weight~$k$. Since the circuit $\bigvee_{k = 0}^n D_k(f)$ is $\dDNNF$ representing $f$, there must be a function $\kappa : \calF \rightarrow \mathbb{N}$ such that $|D_{\kappa(f)}(f)| = 2^{\Omega(\sqrt{|var(f)|})}$ holds for all $f$. Define the class $\calF^* = \{f \land [w(\cdot) = \kappa(f)] \mid f \in \calF\}$. 

We claim that in any $\wDNNF$ representing a satisfiable function in $\calF^*$, all term subcircuits have the same variables. Consider a $\wDNNF$ representing $f \land [w(\cdot) = k]$. Let $T$ be one of its term subcircuit and assume $\var(f) \setminus \var(T) \neq \emptyset$. Let $x \in \var(f) \setminus \var(T)$, $T$ has a model $a$ with $x$ set to 0 and another model $a'$ identical to $a$ but with $x$ set to 1. But $w(a) \neq w(a')$ so $a$ and $a'$ cannot both satisfy $f \land [w(\cdot) = k]$, a contradiction. So all term subcircuits contain all variables.

Combining the above and
Proposition~\ref{proposition:proof_trees_on_same_variables} for $\calF^*$, we get that there is a polynomial relating  the $\dDNNF$- and $\dwDNNF$-sizes of functions of $\calF^*$. So all functions of $\calF^*$ have exponential $\dwDNNF$-size. Since DNNF support polynomial time restriction to models of fixed-weight -- see for instance the proof of~\cite[Proposition 4.1]{AmarilliBJM17} which can easily be adapted to DNNF -- the functions in $\calF^*$ also have polynomial $\DNNF$-size. So the class $\calF^*$ gives us $\dwDNNF \nleq \DNNF$.
\end{proof}

\begin{lemma}
$\DNNF \nleq \dwDNNF$.
\end{lemma}
\begin{proof}
We consider the class $\calF$ of monotone 2-CNF used in~\cite{BovaCMS14} to prove $\DNNF \nleq \CNF$. Let $F$ be a monotone 2-CNF from $\calF$ on $n$ variables $x_1,\dots,x_n$, $F = \bigwedge_{k = 1}^m (x_{k_0} \lor x_{k_1})$. The size of $F$ is polynomial in $n$ while its equivalent DNNF have size exponential in $n$. Now consider $m$ fresh variables $Z = \{z_1,\dots,z_m\}$ and define $F' = \bigwedge_{k = 1}^m ((\neg z_k \land x_{k_0}) \lor (z_k \land x_{k_1}))$. $F'$ is a $\dwDNNF$, and $\exists Z. F' \equiv F$ ($F$ equals $F'$ after forgetting variables $Z$, see~\cite{DarwicheM02} if needed). Since DNNF support polynomial time variables forgetting~\cite{DarwicheM02}, DNNF circuits equivalent to $F'$ have exponential size. Thus the class of the circuits $\{F' | F \in \calF \}$ proves the separation $\DNNF \nleq \dwDNNF$.
\end{proof}

\begin{lemma} $\wDNNF < \dwDNNF$.
\end{lemma}
\begin{proof}
$\dwDNNF \subset \wDNNF$ implies $\wDNNF \leq \dwDNNF$. For $\dwDNNF \nleq \wDNNF$ observe that otherwise we would have $\dwDNNF \simeq \wDNNF$, which would imply $\dwDNNF \leq \DNNF$, thus contradicting Lemma~\ref{lemma:dednnfvsdnnf}.
\end{proof}

\begin{lemma} $\dwDNNF < \dDNNF$ and $\dwDNNF < \sdwDNNF$.
\end{lemma}
\begin{proof}
$\dDNNF$ is a subclass of $\dwDNNF$ so $\dwDNNF \leq \dDNNF$. And $\dDNNF \nleq \dwDNNF$ holds for otherwise $\DNNF \nleq \dwDNNF$ would be violated by transitivity.

$\sdwDNNF$ is a subclass of $\dwDNNF$ so $\dwDNNF \leq \sdwDNNF$. Since $\sdwDNNF$, $\sdDNNF$ and $\dDNNF$ are equally succinct, there must be $\sdwDNNF \nleq \dwDNNF$ otherwise $\dDNNF \nleq \dwDNNF$ would be violated by transitivity.
\end{proof}

This last lemma finishes the proof of Theorem~\ref{theorem:succinctness_map}.

\section{Beginning the Map for Positive AC}

In this section, we will start drawing a succinctness map for positive AC. Recall that positive AC compute non-negative functions but allow for negative constant inputs or equivalently subtraction. It is known that adding subtraction to arithmetic circuits can decrease their size exponentially~\cite{Valiant80}, so $\posAC < \monAC$.

Since there is no apparent mapping between positive AC and a class of Boolean circuits similar to the mapping $\phi$ introduced in Section~\ref{section:phi}, we do not obtain a succinctness map for positive AC in the same way we did for monotone AC.
We here solve some of the relations on the corresponding map, leaving its completion for future work.

\begin{lemma}\label{lemma:deterministic_positive}
Let $C$ be a (smooth) (weakly) decomposable deterministic positive $\AC$. Switching the signs of all negative constants in $C$ yields an equivalent (smooth) (weakly) decomposable deterministic monotone $\AC$. Therefore the relation $\textup{d-}\gamma\textup{-}\posAC \simeq \textup{d-}\gamma\textup{-}\monAC$ holds for any $\gamma \in \{\textup{D,wD,sD,swD}\}$.
\end{lemma}
\begin{proof}
Smoothness and (weak) decomposability are clearly preserved by the transformation. Recall that no two term subcircuits of $C$ can compute an non-zero value on the same assignment, and that the sum of the functions they compute is that computed by~$C$. So each term subcircuit~$T$ computes a positive function, therefore the negative constants in~$T$ must be in even number. But then switching the signs of negative constants in $C$ does not change the function computed by any term subcircuit. Thus the monotone AC we get is equivalent to $C$ and, since its term subcircuits still have pairwise disjoint support, it is deterministic.
\end{proof}

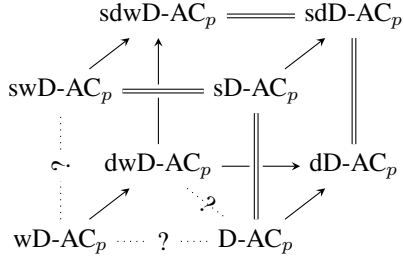
\begin{figure}[t]
\centering
\begin{tikzpicture}
\def\x{2.6};
\def\y{2};
\def\dx{0.5*\x};
\def\dy{0.5*\y};

\node (wD) at (0,0) {$\wDposAC$};
\node (D) at (\x,0) {$\DposAC$};
\node (dwD) at (\dx,\dy) {$\dwDposAC$};
\node (dD) at (\x+\dx,\dy) {$\dDposAC$};
\node (swD) at (0,\y) {$\swDposAC$};
\node (sdwD) at (\dx,\y+\dy) {$\sdwDposAC$};
\node (sdD) at (\x+\dx,\y+\dy) {$\sdDposAC$};
\node (sD) at (\x,\y) {$\sDposAC$};

\draw[-stealth] (dwD) -- (dD);
\draw[-stealth] (dwD) -- (sdwD);
\draw[double distance=1pt] (dD) -- (sdD);
\draw[double distance=1pt] (sdwD) -- (sdD);

\draw[-stealth] (swD) -- (sdwD);
\draw[-stealth] (sD) -- (sdD);
\draw[white,line width=5pt] (swD) -- (sD);
\draw[double distance=1pt] (swD) -- (sD);

\draw[dotted] (wD) -- node[fill=white, midway, rotate=90] {?} (swD);
\draw[-stealth] (wD) -- (dwD);
\draw[white,line width=5pt] (D) -- (sD);
\draw[double distance=1pt] (D) -- (sD);
\draw[-stealth] (D) -- (dD);
\draw[dotted] (wD) -- node[fill=white, midway] {?} (D);
\draw[dotted] (dwD) -- node[midway, rotate=-45] {?} (D);

\end{tikzpicture}
\caption{Partial succinctness map for subclasses of positive AC}\label{fig:map_positive_AC}
\end{figure}

\begin{lemma}
$\gamma\textup{-}\posAC < \textup{d-}\gamma\textup{-}\posAC$ for any $\gamma \in \{\textup{D,w,sD,swD}\}$.
\end{lemma}
\begin{proof}
Monotone AC are positive AC so $\gamma\textup{-}\posAC \leq \gamma\textup{-}\monAC$. Using Lemma~\ref{lemma:deterministic_positive} and Theorem~\ref{theorem:succinctness_map}, we get $\gamma\textup{-}\posAC \leq \gamma\textup{-}\monAC < \textup{d-}\gamma\textup{-}\monAC \simeq \textup{d-}\gamma\textup{-}\posAC$ and hence the result.
\end{proof}

\begin{lemma}
$\DposAC \simeq \sDposAC \simeq \swDposAC$.
\end{lemma}
\begin{proof}[Proof sketch] The algorithm of~\cite{PeharzTPD15} works on AC with non-negative constants but remains sound (with no change) when negative constants are allowed. So smooth weakly decomposable positive AC can be made smooth and decomposable in polynomial time, hence $\sDposAC \leq \swDposAC$. Since $\sDposAC$ is also a subclass of $\swDposAC$, the second succinctness equivalence holds.

$\DposAC \simeq \sDposAC$: there only is $\sDposAC \leq \DposAC$ to prove. Let $C \in \DposAC$, if $g \in C$ is a $+$-node such that $x \in \var(g_r)$ and $x \not\in \var(g_l)$, then add a $\times$-node between $g_l$ and $g$ whose successors are $g_l$ and \raisebox{-0.5\height}{
\begin{tikzpicture}
\node[inner sep=1,font=\footnotesize] (g) at (0,0) {$+$};
\node[inner sep=1,font=\footnotesize] (l) at (-0.4,-0.5) {$x$};
\node[inner sep=1,font=\footnotesize] (r) at (0.4,-0.5) {$\overline{x}$};

\draw (g) -- (l);
\draw (g) -- (r);
\end{tikzpicture}
}
This does not impact decomposability. Inserting $\times$-nodes this way for each non smooth $+$-nodes yields a smooth decomposable AC equivalent to $C$ of size at most $O(|\var(C)| \times |C|)$.
\end{proof}

The above lemmas are summarized in Figure~\ref{fig:map_positive_AC}. Three relations, indicated by question marks in the figure are open. Note that the known relations between classes of positive AC coincide with the corresponding relations between classes of monotone AC,  which motivates the following question:

\begin{question}\label{question:one}
Do all succinctness relations between classes of monotone $\AC$ shown Figure~\ref{fig:map_monotone_AC} hold for the corresponding classes of positive $\AC$ as well?
\end{question} 

Note that completing the map for positive $\AC$ might be very hard: in fact, it would in particular require showing strong lower bounds for $\DposAC$, a well-known open problem in complexity theory for which the best current result is a recent nearly quadratic lower bound~\cite{AlonKV20}. Another question is the relations between the map of monotone $\AC$ and that of positive $\AC$: when imposing determinism, the expressive power of positive $\AC$ is exactly that of monotone $\AC$, while for unrestricted circuits it is known that positive AC are more succinct than monotone AC~\cite{Valiant80}.

\begin{question}\label{question:two}
For which $\gamma \in \{\textup{D,wD,sD,swD}\}$ do we have $\gamma\textup{-}\posAC < \gamma\textup{-}\monAC$?
\end{question} 

\section{Lower Bounds for Positive $\AC$}

\subsection{Sum of Decomposable Products}

In this section we describe a technique to show lower bounds on the size of \emph{structured}-decomposable positive AC. For NNF, structured decomposability is defined with help of a v-tree (variable tree)~\cite{PipatsrisawatD08} but the definition usually assumes that constant inputs have been propagated away in the circuit. This is impossible in our model, so we use the v-tree-free definition from~\cite{VergariCLTB21}. The definition assumes smoothness for simplicity.

\begin{definition}
An AC $C$ is called smooth \emph{structured-decomposable} when it is smooth and decomposable and, for all $Y \subseteq \var(C)$ there is a partition $Y = Y_0 \cup Y_1$ such that, the successors of all $\times$-nodes $g$ in $C$ with $\var(g) = Y$ verify $\var(g_l) = Y_i$ and $\var(g_r) = Y_{1-i}$ for some $i \in \{0,1\}$.
\end{definition}

\begin{definition}
Let $Z$ be a set of variables. A \emph{decomposable product} over $Z$ is a function from $Z$ to $\mathbb{R}$ that can be written as a product $f(X) \times h(Y)$ where $(X,Y)$ is a partition of $Z$ and $f$ and $h$ are functions to $\mathbb{R}$. The decomposable product is called \emph{balanced} when $\frac{|Z|}{3} \leq |X|,|Y| \leq \frac{2|Z|}{3}$.
\end{definition}
A common approach to proving lower bounds for decomposable $\AC$ analyzes representations of the function it computes in terms of sums of balanced decomposable products. Roughly put, the idea is that when more summands are needed in such a representation, $\AC$ for it need to be larger. This technique has been used in recent and not so recent articles, see e.g.~\cite{Valiant80,RazY11,MartensM14}. Translated to Boolean circuits, decomposable products correspond to \emph{combinatorial rectangles}, a tool from communication complexity used in the context of DNNF~\cite{BovaCMS16}.

Variations of the next theorem have been shown several times independently in the literature, see for instance~\cite[Theorem 38]{MartensM14}. The structured case comes from an easy modification of that proof, the rough idea is that each decomposable product is built from a different node of the circuit and, thanks to structuredness, all these nodes have the same set of variables, which eventually yields the same partition for the decomposable products.

\begin{theorem}\label{theorem:sum_of_decomposable_functions} Let $F$ be a non-negative function on $0/1$-variables computed by a decomposable smooth $\AC$ $C$
. Then~$F$ can be written as a sum of $N$ balanced decomposable products over $\var(F)$, with $N \leq |C|$  in the form\footnote{Note that~\cite[Theorem 38]{MartensM14} is stated with $N \leq |C|^2$ because the internal nodes in their AC (or SPN) do not have exactly two successors, as do ours. However, they reduce to AC with that property and the square comes from the quadratic size increase in this reduction.}.
\begin{align*}F = \sum_{i = 1}^N f_i(X_i) \times h_i(Y_i).\end{align*}
If $C$ is structured, the $N$ partitions $(X_i,Y_i)$ are all identical.
\end{theorem}
\subsection{Lower bounds for structured decomposable positive AC}
In this section we prove the following lower bound.
\begin{proposition}\label{proposition:lower_bounds_str_SPN}
There is a class of positive functions $\mathcal{F}$ such that, 
for all $F \in \mathcal{F}$, the smallest AC computing $F$ has size polynomial in $|\var(F)|$ but the smallest smooth structured decomposable AC computing $F$ has size $2^{\Omega(|\var(F)|)}$.
\end{proposition}

By Theorem~\ref{theorem:sum_of_decomposable_functions}, the smallest $N$ for which one can write~$F$ as $F = \sum_{i = 1}^N f_i(X) \times h_i(Y)$ where $f_i(X) \times h_i(Y)$ are balanced decomposable products for the unique partition $(X,Y)$ of $\var(F)$, is a lower bound on the size of all smooth structured decomposable AC computing $F$. Thus, proving Proposition~\ref{proposition:lower_bounds_str_SPN} boils down to finding non-negative functions where the smallest such $N$ depends exponentially on the number of variables.
\par Let us fix a function $F$ and a partition $(X,Y)$. The \emph{value matrix} of $F$ with respect to $(X,Y)$ is a $2^{|X|}\times 2^{|Y|}$ matrix~$M_F$ whose rows (resp. columns) are uniquely indexed by assignments to $X$ (resp. $Y$) and such that, for each pair of indices $(a_X, a_Y)$, the entry of $M_F$ at the $a_X$ row and $a_Y$ column is $F(a_X \cup a_Y)$.

\begin{lemma}\label{lemma:rank_lower_bound} 
Let $F = \sum_{k = 1}^N f_k(X) \times h_k(Y)$ where for all~$k$ we have $f_k \times h_k \neq 0$. Let $M_F$ be the value matrix for $F$ and let $M_i$ denote the the value matrix for $f_i \times h_i$ with respect to partition $(X,Y)$. Then
$$ \rk(M_F) \leq \sum_{k = 1}^N \rk(M_k) = N.$$
\begin{proof}
By construction, $M_F = \sum_{k = 1}^N M_k$, so $\rk(M_F) \leq \sum_{k = 1}^N \rk(M_k)$ holds by sub-additivity of the rank. We now show that $\rk(M_k) = 1$ holds for each $k$. Since $f_k \times h_k \neq 0$, there is a row in $M_k$ which is not a 0-row. Say it is the row indexed by $a_X$. Then the entries in that row are $f_k(a_X)\times h_k(a_Y)$ for varying $a_Y$. In any other rows indexed by $a'_X$, the entries are $f_k(a'_X)\times h_k(a_Y) = (f_k(a'_X) / f_k(a_X))\times f_k(a_X)\times h_k(a_Y)$ for varying $a_Y$. Consequently, all rows are multiples of the $a_X$-row, in other words, all rows of $M_k$ are linearly dependent, hence $\rk(M_k) =~1$.
\end{proof}
\end{lemma}

Using Lemma~\ref{lemma:rank_lower_bound}, one sees that proving Proposition~\ref{proposition:lower_bounds_str_SPN} boils down to finding functions whose value matrices with respect to \emph{any} balanced partition $(X,Y)$ have rank exponential in the number of variables.

The functions we construct are based on graphs. Let $G = (V,E)$ be a graph, denote $n = |V|$ and, for each vertex~$v_i$ in~$V$, create a Boolean variable~$x_i$. 
We consider the function
\begin{equation}\label{eq:one_two_product_function}
F_G(x_1,\dots,x_n) = \prod_{(v_i,v_j) \in E} (1 + \max(x_i,x_j))
\end{equation}
Essentially, for each edge of $G$, if at least one of its end-points is assigned 1 in the assignment, then the edge contributes a factor 2 to the product, otherwise it contributes a factor 1. Regardless of the choice of $G$, the function $F_G$ has~a small positive AC: one just has to write $\max(x_i,x_j) = x_i + x_j - x_ix_j$ and see that the number of $\times$ and $+$ operations needed to compute $F_G$ is polynomial in $n$.

Recall that an \emph{induced matching} is a set $E' \subseteq E$ of edges with pairwise disjoint endpoints, whose set we denote $V'$,  such that all edges of $G$ connecting vertices in $V'$ are in~$E'$. 

\begin{lemma}\label{lemma:determinant_not_null}
Let $F_G$ be as described by~$(\ref{eq:one_two_product_function})$, let $(X,Y)$ be a partition of $var(F_G)$ and $(V_X,V_Y)$ be the corresponding partition of $V$. If there is an induced matching $m$ in $G$ between vertices $V_l$ and $V_r$ such that $V_l \subseteq V_X$ and $V_r \subseteq V_Y$, then 
$$
\rk(M_{F_G}) \geq 2^{|m|}
$$
where $M_{F_G}$ is the value matrix of $F_G$ for the partition $(X,Y)$ and $|m|$ is the number of edges in $m$. 
\begin{proof}
Rename $M := M_{F_G}$. Identify each vertex with its variable in $\var(F_G)$ and let $(x_i,y_i)_{i \in [|m|]}$ be the edges of $m$, with $x_i \in X$ and $y_i \in Y$. Order the variables in $X$ as $X = (x_1,\dots,x_{|X|})$ and the variables in $Y$ as $Y = (y_1,\dots,y_{|Y|})$, so that the $|m|$ first variables in each set correspond to the nodes in the matching. Permutations of rows or columns do not change the rank of a matrix so we assume that the assignments indexing the rows and the columns are ordered so that, when seeing the assignments as tuples of $0$ and $1$, the integers encoded in binary by the tuples are ordered. More formally $a_X$ is before $a'_X$ if and only if $\sum_{k} a(x_k)2^{k-1} < \sum_{k} a'(x_k)2^{k-1}$. Now consider all $2^{2|m|}$ truth assignments to $var(F_G)$ where variables corresponding to vertices not in $V_l \cup V_r$ are set to~$0$. Let $M^*$ be the $2^{|m|} \times 2^{|m|}$ sub-matrix of $M$ obtained by keeping only rows and columns indexed by these assignments. The rank of a sub-matrix is always at most that of the matrix, so $\rk(M^*) \leq \rk(M)$.
To prove the lemma, it is enough to show that $\rk(M^*) = 2^{|m|}$, which holds if and only if $\det(M^*) \neq 0$. 
For $0 \leq i \leq |m|$, let $M^*_i$ be the matrix containing the first $2^i$ rows and first $2^i$ columns of $M^*$. We prove by induction on~$i$ that all $M^*_i$ have non-zero determinant, which will prove that $M^*$ (which is $M^*_{|m|}$) has non-zero determinant, and therefore full rank.
For the base case, $M^*_0 = (1)$ has determinant~$1$.
For the general case, assume that $\det(M^*_i) \neq 0$ and observe that $M^*_{i+1} =  \left(\begin{array}{c|c} M^*_i  & 2M^*_i \\ \hline  2M^*_i &  2M^*_i \end{array}\right)$. The determinant of $M^*_{i+1}$ is
\begin{align*}\det\left(\begin{array}{c|c} M^*_i  & 2M^*_i \\ \hline  2M^*_i &  2M^*_i \end{array}\right) 
= \det\left(\begin{array}{c|c} -M^*_i  & 2M^*_i \\ \hline  0 &  2M^*_i \end{array}\right)\\
= \det(-M^*_i)\det(2M^*_i) = (-2)^{2^i} \det(M^*_i)^2 \neq 0.
\end{align*}
\end{proof}
\end{lemma}

So if, for \emph{every} balanced partition of $V$, we have a large enough induced matching $M$ between the two sides, then the rank of the value matrix for $F_G$ for any balanced partition is large, thus many balanced decomposable products are needed in a sum representing $F_G$. The only thing left is to find graphs $G$ with this ``large enough matching'' property, which turn out to be \emph{expander graphs}. A $d$-regular graph is a graph whose vertices all have degree $d$. A $(c,d)$-expander graph on vertices $V$ is a $d$-regular graph such that for any $S \subseteq V$ of size $|S| \leq |V|/2$, it holds that $|N(S)| \geq c|S|$, where $N(S) = \{v \in V \setminus S \mid (u,v) \in E, u \in S\}$.

\begin{theorem}\cite[Section 9.2]{AlonS00}\label{theorem:expander_graphs}
There is, for some $c > 0$, an infinite sequence of $(c,3)$-expander graphs $(G_i)_{i \in \mathbb{N}}$.
\end{theorem}

\noindent We use these expander graphs for our lower bound. 

\begin{lemma}\label{lemma:matchings_in_expander_graph}
Let $G = (V,E)$ be a $(c,3)$-expander graph with $n = |V|$, and let $V = V_1 \uplus V_2$ be a balanced partition of~$V$, then there exists an induced matching $m$ of size $\Omega(n)$ between $V_1$ and~$V_2$.
\begin{proof}
$V_1$ or $V_2$ has size at most $n/2$, say $|V_1| \leq n/2$. There is $N(V_1) \subseteq V_2$ and $|N(V_1)| \geq c|V_1| \geq cn/3$ where the last inequality comes from the partition being balanced. So at least $cn/3$ edges connect $V_1$ to $V_2$. Since $G$ is 3-regular, at least a third of these edges form an matching in $G$, and a third of these matching edge share no endpoint in $V_1$, and finally a third of these edges share no endpoint in $V_2$ either. So we obtain a induced matching between $V_1$ and $V_2$ of size at least $cn/81$.
\end{proof}
\end{lemma}

Combining Theorems~\ref{theorem:sum_of_decomposable_functions} and~\ref{theorem:expander_graphs} with Lemmas~\ref{lemma:rank_lower_bound}, ~\ref{lemma:determinant_not_null} and ~\ref{lemma:matchings_in_expander_graph} yields Proposition~\ref{proposition:lower_bounds_str_SPN}.

\section{Conclusion}

We have started drawing succinctness maps for arithmetic circuits modeled after that proposed for $\NNF$ in~\cite{DarwicheM02}. Due to great amount of recent work on practical applications of AC with specific structural restrictions, we have studied classes of AC for combinations of four key restrictions. Using a mapping between monotone AC and NNF, we have drawn the full succinctness map for monotone AC by lifting the existing map for NNF and extending it to incorporate new classes defined with weak decomposability. In certain cases we could show that positive and monotone AC have the same expressive power, which gave us some succinctness results between classes of positive AC for free. We leave the challenging task of determining the remaining relations between classes of positive AC as an open question. 
Finally, we have also introduced techniques to prove lower bounds on structured positive AC and applied them to the case of smooth structured-decomposable~AC.

\appendix
\section*{Appendix}

\setcounter{lemma}{0}

\begin{lemma}
Let $C$ be a (weakly) decomposable $\AC$ (resp.~$\NNF$) then:
\begin{itemize}
\item if $C$ is smooth, all variables appear in all term subcircuits
\item if $C$ is deterministic, then any two distinct term subcircuits $T$ and $T'$ verify $T \times T' = 0$ (resp. $T \land T' \equiv 0$).
\end{itemize}
\begin{proof}
We only show the lemma for AC, as the proof for NNF is completely analogous. The two points are shown by induction on the depth of $C$, i.e., the number of nodes in~a longest directed path in $C$. AC of depth 1 are single variable nodes or constant nodes, and thus the statement of the lemma is straightforward. Now assume the lemma holds for all (w)D-AC of depth at most $k$ and consider an (w)D-AC $C$ of depth $k+1$. Let $g$ be the source node of $C$. Let $C_l$ and $C_r$ be the (w)D-AC under $g_l$ and $g_r$.

If $g$ is a $\times$-node then the term subcircuits of $C$ are products $T = T_l \times T_r$ where $T_l$ and $T_r$ are term subcircuits of $C_l$ and $C_r$. If $C$ is smooth, then $\var(T) = \var(T_l) \cup \var(T_r) = \var(C_l) \cup \var(C_r) = \var(C)$ holds by induction. If $C$ is deterministic, then let $T = T_l \times T_r$ and $T' = T'_l \times  T'_r$ be distinct term subcircuits of $C$. We have $T'_l \neq T_l$ or $T'_r \neq T_r$ and thus by induction, for every assignment $a$, we have $T_l(a) \times  T'_l(a) = 0$ or $T_r(a) \times  T'_r(a) = 0$, so $T' \times T = 0$.

If $g$ is a $+$-node, then every term subcircuit $T$ of $C$ is either equivalent to a term subcircuit $T_l$ of $C_l$ or to a term subcircuit $T_r$ of $C_r$. Assume $C$ is smooth, then $\var(C) = \var(C_l) = \var(C_r)$, but then $\var(T)$ is either $\var(T_l)$ or $\var(T_r)$, which by induction equals $\var(C)$. Now when $C$ is deterministic there is $C_l \times C_r = 0$, so any term subcircuits $T = T_l$ and $T' = T'_r$ verify $T \times T' = 0$. And by induction any two distinct subcircuits $T = T_l$ and $T' = T'_l$ verify $T \times T' = T_l \times T'_l = 0$ (likewise for $T_r$ and $T'_r$).
\end{proof}
\end{lemma}

\begin{lemma}
When $C$ is a monotone $\AC$, $\phi(C)$ is an $\NNF$ whose models are $\supp(C)$. Moreover if $C$ is (weakly) decomposable, deterministic, or smooth, then $\phi(C)$ is as well.
\end{lemma}
\begin{proof}
The graph of $\phi(C)$ is that of $C$ and $\phi$ contains only $\land$- and $\lor$-nodes, thus $\phi(C)$ is an NNF. It is easy to see that for each node $g$ in $C$ we have $\var(g) = \var(\phi(g))$, so smoothness and (weak) decomposability are preserved. 

We prove that by induction on the depth of $C$ that (1) $\sat(\phi(C)) = \supp(C)$ and (2) if $C$ is deterministic, then so is $\phi(C)$. If $C$ has depth 1, then it is either a constant input or a literal input. In the case $C = \alpha$, if $\alpha > 0$ then $\supp(C) = \{a_\emptyset\} = \sat(1) = \sat(\phi(C))$. If $\alpha = 0$ then $\supp(C) = \emptyset = \sat(0) = \sat(\phi(C))$. In the case $C = \ell_x$ there is $\phi(C) = C$ so we are done. Now assume (1) and (2) hold for all AC of depth at most $k$ and suppose $C$ has depth $k+1$. Let $g$ be its source node.

If $g$ is a $\times$-node, then $C(a) = 0$ if and only if $g_l(a_l) = 0$ or $g_r(a_r) = 0$, where $a_l$ and $a_r$ denote the restrictions of $a$ to $\var(g_l)$ and $\var(g_r)$ respectively. So $a \not\in \supp(C)$ if and only if $a_l \not\in \supp(g_l)$ or $a_r \not\in \supp(g_r)$. By induction $\supp(g_{l/r}) = \sat(\phi(g_{l/r})$, so $a \not\in \supp(C)$ if and only if $a \not\in \sat(\phi(g_l) \land \phi(g_r)) = \sat(\phi(C))$. So (1) holds.

If $g$ is a $+$-node, then $C(a) = 0$ iff $g_l(a_l) = 0$ and $g_r(a_r) = 0$. So $a \not\in \supp(C)$ iff $a_l \not\in \supp(g_l)$ and $a_r \not\in \supp(g_r)$. By induction $\supp(g_{l/r}) = \sat(\phi(g_{l/r})$, so $a \not\in \supp(C)$ iff $a \not\in \sat(\phi(g_l) \lor \phi(g_r)) = \sat(\phi(C))$. So (1) holds. As for (2), if $a_r \in \supp(g_r)$ implies $a_l \not\in \supp(g_l)$ and vice-versa, then $\supp(g_{l/r}) = \sat(\phi(g_{l/r})$ yields that the source $\lor$-node of $\phi(C)$ is deterministic.
\end{proof}

\begin{lemma}
For every $\NNF$ $D$, there exists an $\AC$ $C$ of size $|D|$ with $\supp(C) = \sat(D)$. Moreover if $D$ is (weakly) decomposable, deterministic, or smooth, then so is $C$.
\end{lemma}
\begin{proof}
It suffices to replace each $\land$-node in $D$ by a $\times$-node and each $\lor$-node by $+$-node. Let $\psi(D)$ be that AC. Clearly $|D| = |\psi(D)|$, and it is easy to see that for each node $g$ in $D$, $\var(g) = \var(\psi(g))$, so smoothness and (weak) decomposability are preserved by $\psi$. Moreover $\phi(\psi(D)) = D$, so $\sat(D) = \supp(\psi(D))$. Determinism is preserved since $\sat(g) = \supp(\psi(g))$ holds for all nodes $g$ in $D$.
\end{proof}

\section*{Acknowledgment}

The authors were supported by the PING/ACK project of the French National Agency for Research (ANR-18-CE40-0011).

\bibliographystyle{kr}
\bibliography{main}

\end{document}